\documentclass{scrartcl}

\usepackage{soul}
\usepackage{url}
\usepackage[utf8]{inputenc}
\usepackage[small]{caption}
\usepackage{graphicx}
\usepackage{amsmath}
\usepackage{amsthm}
\usepackage{booktabs}
\usepackage{algorithm}
\usepackage{algorithmic}

\usepackage{wrapfig}
\usepackage{booktabs}
\usepackage{caption}
\usepackage{subcaption}
\usepackage{amsmath,amsfonts,amssymb}
\usepackage{bbm}
\usepackage{stmaryrd}

\usepackage{tikz}
\usetikzlibrary{automata, positioning, calc, shapes, arrows, fit}

\usepackage{turnstile}
\usepackage{mathrsfs}
\usepackage{amsthm}
\usepackage{nicefrac}
\usepackage{natbib}
\usepackage{dsfont}

\usepackage{xcolor}

\definecolor{my-red}{HTML}{C62828}
\definecolor{my-red-light}{HTML}{E57373}
\definecolor{my-red-verylight}{HTML}{FFCDD2}
\definecolor{my-red-dark}{HTML}{B71C1C}

\definecolor{my-pink}{HTML}{EC407A}
\definecolor{my-pink-light}{HTML}{F48FB1}
\definecolor{my-pink-verylight}{HTML}{F8BBD0}
\definecolor{my-pink-dark}{HTML}{880E4F}

\definecolor{my-purple}{HTML}{8E24AA}
\definecolor{my-purple-light}{HTML}{BA68C8}
\definecolor{my-purple-verylight}{HTML}{e5cefc}
\definecolor{my-purple-dark}{HTML}{6A1B9A}

\definecolor{my-indigo}{HTML}{3949AB}
\definecolor{my-indigo-light}{HTML}{7986CB}
\definecolor{my-indigo-verylight}{HTML}{9FA8DA}
\definecolor{my-indigo-dark}{HTML}{1A237E}

\definecolor{my-blue}{HTML}{1E88E5}
\definecolor{my-blue-light}{HTML}{64B5F6}
\definecolor{my-blue-verylight}{HTML}{B3E5FC}
\definecolor{my-blue-dark}{HTML}{0D47A1}

\definecolor{my-cyan}{HTML}{00BCD4}
\definecolor{my-cyan-light}{HTML}{4DD0E1}
\definecolor{my-cyan-verylight}{HTML}{80DEEA}
\definecolor{my-cyan-dark}{HTML}{0097A7}

\definecolor{my-teal}{HTML}{009688}
\definecolor{my-teal-light}{HTML}{4DB6AC}
\definecolor{my-teal-verylight}{HTML}{B2DFDB}
\definecolor{my-teal-dark}{HTML}{00695C}

\definecolor{my-green}{HTML}{39ac39}
\definecolor{my-green-light}{HTML}{8cd98c}
\definecolor{my-green-verylight}{HTML}{b3e6b3}
\definecolor{my-green-dark}{HTML}{339933}

\definecolor{my-grass}{HTML}{689F38}
\definecolor{my-grass-light}{HTML}{8BC34A}
\definecolor{my-grass-verylight}{HTML}{AED581}
\definecolor{my-grass-dark}{HTML}{33691E}

\definecolor{my-lime}{HTML}{CDDC39}
\definecolor{my-lime-light}{HTML}{DCE775}
\definecolor{my-lime-verylight}{HTML}{E6EE9C}
\definecolor{my-lime-dark}{HTML}{AFB42B}

\definecolor{my-yellow}{HTML}{fffc29}
\definecolor{my-yellow-light}{HTML}{fffd7a}
\definecolor{my-yellow-verylight}{HTML}{fefdbb}
\definecolor{my-yellow-dark}{HTML}{FFD600}

\definecolor{my-orange}{HTML}{FF8F00}
\definecolor{my-orange-light}{HTML}{FFC107}
\definecolor{my-orange-verylight}{HTML}{ffe5a4}
\definecolor{my-orange-dark}{HTML}{FF6F00}

\definecolor{my-brown}{HTML}{6D4C41}
\definecolor{my-brown-light}{HTML}{795548}
\definecolor{my-brown-verylight}{HTML}{BCAAA4}
\definecolor{my-brown-dark}{HTML}{3E2723}

\definecolor{my-gray}{HTML}{616161}
\definecolor{my-gray-light}{HTML}{9E9E9E}
\definecolor{my-gray-verylight}{HTML}{f0f0f0}
\definecolor{my-gray-dark}{HTML}{424242}

\definecolor{my-steel}{HTML}{546E7A}
\definecolor{my-steel-light}{HTML}{78909C}
\definecolor{my-steel-verylight}{HTML}{B0BEC5}
\definecolor{my-steel-dark}{HTML}{37474F}

\usepackage[unicode=true, bookmarks=false, breaklinks=true, pdfborder={0 0 1}, colorlinks=true]{hyperref}
\hypersetup{linkcolor=my-blue-dark, citecolor=my-blue-dark, filecolor=my-blue, urlcolor=my-blue-dark}

\newtheorem{theorem}{Theorem}
\newtheorem{lemma}[theorem]{Lemma}
\newtheorem{corollary}[theorem]{Corollary}

\newtheorem{proposition}[theorem]{Proposition}
\newtheorem{definition}{Definition}

\newtheorem{myex}{Example}

\newcommand{\tuple}[1]{\left\langle #1 \right\rangle}

\DeclareMathOperator*{\argmax}{argmax}

\renewcommand{\phi}{\varphi}

\DeclareMathOperator{\concat}{\oplus}

\newcommand{\projSet}{\mathcal{P}}
\newcommand{\allocSet}{\mathcal{A}}
\newcommand{\allocSetEx}{\mathcal{A}_{\mathit{EX}}}
\newcommand{\profile}{\boldsymbol{A}}
\newcommand{\pbRule}{F}
\newcommand{\truth}{\pi^\star}
\newcommand{\noiseModel}{\mathcal{M}}
\newcommand{\noiseMmaxapp}{\noiseModel_{\mathit{app}}}
\newcommand{\noiseMmaxappNORM}[1]{Z_{#1}^{\mathit{app}}}
\newcommand{\noiseMmaxNashAppCost}{\noiseModel_{\mathit{Ncost}}}
\newcommand{\noiseMmaxNashAppCostNORM}[1]{Z_{#1}^{\mathit{Ncost}}}
\newcommand{\noiseMmaxNashApp}{\noiseModel_{\mathit{Napp}}}
\newcommand{\noiseMmaxNashAppNORM}[1]{Z_{#1}^{\mathit{Napp}}}
\newcommand{\prob}{\mathbb{P}}
\newcommand{\greedyApp}{\pbRule_{\mathit{greed}}}

\newcommand{\appMaximizing}{\pbRule_{\mathit{app}}}
\newcommand{\costAppMaximizing}{\pbRule_{\mathit{cost}}}

\newcommand{\appNashMaximizing}{\pbRule^N_{\mathit{app}}}
\newcommand{\costAppNashMaximizing}{\pbRule^N_{\mathit{cost}}}
\newcommand{\relAppNashMaximizing}{\tilde{\pbRule}^N_{\mathit{app}}}
\newcommand{\relCostAppNashMaximizing}{\tilde{\pbRule}^N_{\mathit{cost}}}

\definecolor{MyOrange}{rgb}{0.88, 0.43, 0.08}
\definecolor{MyGreen}{rgb}{0.28, 0.78, 0.28}

\usepackage{pifont}
\definecolor{MyRed}{rgb}{0.6, 0, 0}
\definecolor{MyGreen}{rgb}{0, 0.6, 0}
\newcommand{\cmark}{{\color{MyGreen}\ding{51}}}
\newcommand{\xmark}{{\color{MyRed}\ding{55}}}

\title{Epistemic Selection of Costly Alternatives: The Case of Participatory Budgeting}
\author{
	Simon Rey \text{\normalfont and} Ulle Endriss\\[0.5em]
	s.j.rey@uva.nl \text{\normalfont and} u.endriss@uva.nl
}
\date{Institute for Logic, Language and Computation (ILLC)\\[0.25em]University of Amsterdam}

\begin{document}
	\maketitle
	
	\begin{abstract}
		We initiate the study of voting rules for participatory budgeting using the so-called epistemic approach, where one interprets votes as noisy reflections of some ground truth regarding the objectively best set of projects to fund. Using this approach, we first show that both the most studied rules in the literature and the most widely used rule in practice cannot be justified on epistemic grounds: they cannot be interpreted as maximum likelihood estimators, whatever assumptions we make about the accuracy of voters. Focusing then on welfare-maximising rules, we obtain both positive and negative results regarding epistemic guarantees.
	\end{abstract}

	\section{Introduction}
	
	The term \emph{participatory budgeting} (PB) has been used to describe a range of mechanisms that directly involve citizens in public spending decisions \citep{Caba04}. The basic idea is that people can vote on grassroots projects (e.g., building a playground or funding a cultural event), each of which has a certain cost. The most popular projects---that fit a given budget constraint---then get funded.  In recent years, PB has flourished around the world, making it one of the most popular tools of participatory democracy. At the same time, it also has received a lot of attention in the literature on (computational) social choice \citep{AzSh20, ReMa23}. 
	
	Given the votes of the citizens, it is not always obvious how to decide which projects to fund (in other words, there are many different voting rules one could consider using). The common approach to choosing a voting rule is the \emph{normative approach}, where we encode normatively desirable properties of voting rules, e.g., properties related to fairness, in the form of axioms and then analyse to what extent a given voting rule will satisfy those axioms. In this paper, we instead explore what is known as the \emph{epistemic approach}---or \emph{truth-tracking approach}---to the analysis and design of voting rules for PB. Both approaches play a central role in the broader field of computational social choice \citep{HBCOMSOC2016}. But in the specific context of PB, the epistemic approach may, at first sight, seem less natural a choice. So let us briefly outline why we nonetheless believe that it is important to explore its potential. 
	
	One reason is that, even in the familiar setting of PB with citizens voting for projects they personally would like to receive public funding, we might interpret these votes as evidence for objective quality. Indeed, whether a given project is a success will often become clear only some time \emph{after} it has been realised: Will local citizens actually use the compost bins? Will the number of mosquitoes go down? Will the new speed camera reduce the number of accidents?\footnote{All these examples are taken from projects that were brought to the vote in Toulouse 2019; see \href{https://jeparticipe.metropole.toulouse.fr/processes/bp2019}{jeparticipe.metropole.toulouse.fr/processes/bp2019} for more details, and in particular the file ``Catalogue des 30 idées soumises au vote''. The data is also hosted on the website \href{http://pabulib.org/?city=Toulouse}{pabulib.org} \citep{FaliszewskiEtAlIJCAI2023}, though the description of the projects is not available there.} So we might think of the citizens casting their votes as agents with bounded rationality who enjoy a noisy view of this \emph{ground truth}. They do not know what the best set of projects to fund is, but each of them is more likely to vote for a good rather than a bad set of projects.
	
	Furthermore, PB as a means to select projects to receive public funding is but one example of a selection process of \emph{costly alternatives}.
	For some other such processes---that are, at least mathematically, equivalent to PB---the existence of a ground truth may be more obvious. Consider, for instance, the case of the \href{https://eternagame.org}{Eterna platform}.\footnote{See \href{https://eternagame.org}{eternagame.org} and \href{https://wikipedia.org/wiki/EteRNA}{wikipedia.org/wiki/EteRNA} for more information.} On this collaborative platform, users can submit different ways of folding a given protein. A subset of the proposed configurations is then synthesised in a laboratory to determine which are most stable. One can think of this as a PB process: the projects are the different protein foldings; their cost is the cost of synthesising them; the budget limit is the amount of money that is allocated to this process; and finally, the protein foldings submitted by a user constitute their approval ballot. Mathematically speaking, this is thus a well-defined PB process. Moreover, there is a clear ground truth here: a set of objectively most stable protein configurations. Interestingly, this is also the motivating example of the first epistemic analysis of multi-winner voting rules,\footnote{In multi-winner voting with approval ballots \citep{LaSk22}, we ask each voter which of the candidates up for election they approve of and then, for some fixed~$k$, need to choose a committee of $k$ candidates. This scenario is mathematically equivalent to PB when all projects have the same cost.} though in a setting without costs \citep{PRS12}.
	
	A last example is that of a selection committee for research grant proposals. In such a committee, the members have to decide which of the grants should be funded. The grants typically have different costs, and there is a maximum amount of money that can be allocated. Then one might argue that the proposals have a ``ground truth'' probability of success that can only be observed \textit{a posteriori}. At selection time, the committee members observe noisy signals regarding this ground truth through the reports of reviewers, and they need to make a decision based on this information.
	
	If we have a clear idea how these noisy views on the ground truth are generated (votes that are cast in the context of PB, protein configurations that are proposed for the EterRNA platform, or approvals submitted for some grant proposals)---that is, if we have a well-defined \emph{noise model}---we, in principle, are able to design a voting rule that maximises the likelihood of returning the ground truth, i.e., the best set of projects that fit our budget.
	Of course, in practice we do not have access to this noise model. Still, if a natural voting rule turns out to be such a \emph{maximum likelihood estimator} (MLE) for a natural choice of noise model, then we can interpret this as an argument for using that rule. Similarly, if we can prove that for a given rule there does not exist \emph{any} noise model that would make that rule an MLE, then that rule will be less suitable in contexts where it is reasonable to assume that there exists some form of ground truth regarding the right decision to be taken.
	
	\paragraph{Contribution.}
	We first analyse the rule overwhelmingly used in practice around the world---the greedy cost approval rule---as well as the rules most studied in the recent literature---the Method of Equal Shares \citep{PPS21, BFLMP23} and the Sequential Phragmén Rule \citep{LCG22}. Using a necessary condition provided by \citet{CoSa05}, we prove that none of these rules can be interpreted as an MLE, even for instances where all projects have the same cost (corresponding to multi-winner voting instances).
	
	We then turn to a family of rules that all satisfy the aforementioned necessary condition: additive argmax rules. These can be thought of as welfare-maximising rules. We focus on eight specific rules based on either utilitarian or Nash social welfare. In the case of utilitarian social welfare, we show that it is impossible to find a noise model for which the most natural rules would be MLEs in the general case. For Nash welfare, the picture is brighter, since for two rules we are able to identify noise models under which they are MLEs.
	
	Still, the overall picture pained by our analysis is largely negative, demonstrating that it is difficult to design voting rules for PB with good epistemic performance, certainly if this requirement comes on top of other, particularly common normative, requirements. Identifying conditions under which more positive results can be achieved---and ideally conditions that can be satisfied for specific application scenarios where the epistemic perspective is most relevant---therefore presents itself as an important challenge for future work in the field.
	
	\paragraph{Related work.}
	The study of PB rules is part of social choice theory, which more generally deals with the design and analysis of voting rules for different kinds of scenarios. Given that for every application domain many different rules can be---and have been---devised, it can be hard for the decision maker to select the rule to be used. To assist in this choice, two main approaches have been developed. The first one, the \emph{normative} or \emph{axiomatic approach} \citep{Arro51, Thom01}, tries to identify voting rules that satisfy certain normative requirements. The second one, the \emph{epistemic approach} \citep{ElSl16, Piva19}, seeks out rules that can recover a ground truth, assuming that the votes received are noisy estimates of that ground truth.  It is the latter approach we follow here.
	
	Formal work on PB to date instead has followed the axiomatic approach, with a special focus on fairness \citep{ALT18, PPS21, HKPP21, LMR21, LCG22, BFLMP23}, incentive compatibility \citep{FGM16, FPPV21, GKSA19, REH21}, and monotonicity requirements \citep{TaFa19, BBS20, REH20}. We refer the reader to the survey by \citet{ReMa23} for more details.
	
	The epistemic approach has been first applied to the standard voting model \citep{Youn95, CoSa05, CPS14}. Later on, other social choice scenarios have been investigated through the epistemic lens, notably multi-winner elections \citep{PRS12, CPS16}, and judgment aggregation \citep{BoRa06, BDP14, TeEn19}. To the best of our knowledge, the only epistemic study of PB is the one section dedicated to the topic by \citet{GKSA19}, though in the context of divisible PB, i.e., when projects can be partially funded. Interestingly, they show that \emph{knapsack voting}---a PB rule that resembles the greedy cost approval rule in the divisible setting---can be interpreted as a maximum likelihood estimator, while we will see that this is not the case for the greedy cost approval rule in the context of \emph{indivisible}~PB.
	
	\section{Preliminaries}
	\label{sec:prelim}
	
	In this section, we recall the standard model of PB,  used in much of the literature \citep[see, e.g.,][]{TaFa19,AzSh20,ReMa23}, and then define what it means for a PB rule to be a maximum likelihood estimator (MLE).
	
	\subsection{Participatory Budgeting}
	
	A PB problem is described by an instance $I = \tuple{\projSet, c, b}$, where $\projSet$ is the set of available \emph{projects}, $c: \projSet \rightarrow \mathbb{N}$ is the \emph{cost function}---mapping any given project $p \in \projSet$ to its cost $c(p) \in \mathbb{N}$---and $b \in \mathbb{N}$ is the \emph{budget limit}. We write $c(P)$ instead of $\sum_{p \in P} c(p)$ for sets of projects $P \subseteq \projSet$. 
	For a given PB instance, we ask several \emph{agents} to each submit an \emph{approval ballot} $A \subseteq \projSet$, resulting in a vector $\profile$ of ballots, one for each agent. Such a vector of approval ballots is called a \emph{profile}. 
	Given two profiles $\profile$ and $\profile'$, we use $\profile \concat \profile'$ to denote the profile obtained by concatenating them.
	
	Given an instance $I = \tuple{\projSet, c, b}$, we need to select a subset of projects $\pi \subseteq \projSet$ to implement. Such a \emph{budget allocation}~$\pi$ has to be \emph{feasible}, i.e., we require $c(\pi) \leq b$. Let $\allocSet(I) = \{\pi \subseteq \projSet \mid c(\pi) \leq b\}$ be the set of feasible budget allocations for~$I$. Moreover, let $\allocSetEx(I)$ be the set of all \emph{exhaustive} budget allocations for~$I$, i.e., of allocations $\pi \in \allocSet(I)$ for which there is no project $p \in \projSet \setminus \pi$ such that $c(\pi \cup \{p\}) \leq b$. 
	
	Computing budget allocations is done by means of \emph{PB rules}. An \emph{irresolute} PB rule $\pbRule$ is a function that takes as input a PB instance $I$ and a profile $\profile$ over $I$, and that returns a nonempty set of feasible budget allocations $\pbRule(I, \profile) \subseteq \allocSet(I)$. A rule is \emph{exhaustive} if $\pbRule(I, \profile) \subseteq \allocSetEx(I)$ for all $I$ and $\profile$.
	
	Some of our results will apply only to unit-cost instances. An instance $I = \tuple{\projSet, c, b}$ is a \emph{unit-cost} instance if there exists an $\ell \in \mathbb{N}$ such that $(i)$~$c(p) = \ell$ for all projects $p \in \projSet$ and $(ii)$~$b \equiv 0 \bmod \ell$.
	This restriction is particularly interesting since unit-cost instances are equivalent to multi-winner voting problems where one needs to elect a committee of some fixed size~$k$ \citep{FaliszewskiEtAlTRENDS2017}. Candidates can be thought of as projects of cost~$\ell$, so under a budget limit of $b=k\cdot\ell$ exhaustive budget allocations correspond to such committees.
	
	\subsection{The Truth-Tracking Perspective}\label{sec:truthtracking}
	
	Under the truth-tracking perspective, we assume that, for every given instance~$I$, there exists an objectively best feasible budget allocation, the so-called \emph{ground truth} $\truth$, and we would like every reasonable rule to select $\truth$. The ground truth is not known, neither to the agents nor to the decision maker. We will thus assess the quality of PB rules based on their ability to retrieve the ground truth given noisy votes.
	
	Formally, a \emph{noise model} $\noiseModel$ is a generative model that produces random approval ballots for a given instance and ground truth. We represent it as a probability distribution over all approval ballots. For a given instance $I = \tuple{\projSet, c, b}$, ground truth $\truth \in \allocSet(I)$, and approval ballot $A \subseteq \projSet$, we denote by $\prob_\noiseModel(A \mid \truth, I)$ the probability for the noise model $\noiseModel$ to generate ballot $A$ given $I$ and $\truth$.
	For profiles generated by the noise model, ballots are drawn identically and independently from $\noiseModel$.
	
	Suppose the noise model $\noiseModel$ indicates how the voters form their preferences for any given ground truth.\footnote{We stress that here we are making the (implicit) assumption that every voter's ballot is generated by the same noise model. This simplifying assumption, though common in the literature \citep{ElSl16}, does not account for the fact that different voters may have different expertise. Specifically, in the context of PB, a voter living in a specific district may be more likely to accurately judge the quality of a project based in the same district than a voter who lives elsewhere.} Then a good rule should select the outcome that most likely would have generated the observed profile we observe if it were the ground truth plugged into~$\noiseModel$. This is the \emph{maximum likelihood estimator} (MLE) for~$\noiseModel$.
	
	\begin{definition}[Maximum likelihood estimators]
		For a noise model $\noiseModel$, the likelihood of a profile $\profile$ over the instance $I$ and a budget allocation $\pi \in \allocSet(I)$ is defined as:
		\[L_{\noiseModel}(\profile, \pi, I) = \prod_{A \in \profile} \prob_\noiseModel(A \mid \pi, I).\]
		A PB rule $\pbRule$ is said to be the MLE for $\noiseModel$, if for every instance $I$ and every profile $\profile$ we have:
		\[\pbRule(I, \profile) = \argmax_{\truth \in \allocSet(I)} L_{\noiseModel}(\profile, \truth, I).\]
	\end{definition}
	
	\noindent
	In the context of the standard model of voting theory, \citet{CoSa05} identified a necessary condition for a voting rule to be interpretable as an MLE: it should satisfy what we are going to call \emph{weak reinforcement}. This result straightforwardly carries over to the PB setting.
	
	\begin{definition}[Weak reinforcement]
		A PB rule $\pbRule$ is said to be satisfying weak reinforcement if and only if, for every instance $I$ and every two profiles  $\profile$ and $\profile'$, we have:
		\[\pbRule(I, \profile) = \pbRule(I, \profile') \;\Longrightarrow\; \pbRule(I, \profile \concat \profile') = \pbRule(I, \profile).\]
	\end{definition}
	
	\begin{lemma}[\citeauthor{CoSa05}, \citeyear{CoSa05}]
		\label{lem:MLENecessaryCondition}
		If a PB rule $\pbRule$ does not satisfy weak reinforcement, then there exists no noise model $\noiseModel$ for which $\pbRule$ is the MLE.
	\end{lemma}
	
	\noindent Note that this result applies for any set of possible ground truths, so also if we assume the ground truth to be exhaustive.
	
	As a first application of this lemma, let us consider what 	perhaps is the most widely used rule in real-world PB elections, namely the \emph{greedy cost approval rule} $\greedyApp$ \citep{GKSA19,AzSh20}. Under this rule, we select projects in order of the number of approvals they receive, skipping over any projects that would lead us to exceed the budget.\footnote{The name of the rule is linked to the fact that it can be thought of as approximating the maximum utilitarian social welfare under the cost satisfaction function \citep{ReMa23}. These concepts will be formally introduced in Section~\ref{subsec:Utilitarian_Social_Welfare}.} 
	%
	%
	%
	%
	Note that $\greedyApp$ is exhaustive.
	%
	Unfortunately, we can easily show that it cannot be interpreted as an MLE.
	
	\begin{proposition}\label{prop:greedy}
		There exists no noise model $\noiseModel$ such that the greedy cost approval rule $\greedyApp$ is the MLE for~$\noiseModel$.
	\end{proposition}
	
	\begin{proof}
		Consider an instance $I$ with three projects denoted by $p_1$, $p_2$ and $p_3$, a budget limit of $b = 4$, and costs as shown in the table below. Moreover, consider two profiles $\profile$ and $\profile'$ in which the approval scores are as presented below.
		\begin{center}
			\begin{tabular}{ccccc}
				\toprule
				& \textbf{Cost} & \begin{tabular}{@{}c@{}}\textbf{Approval score} \\ \textbf{in $\profile$} \end{tabular} & \begin{tabular}{@{}c@{}}\textbf{Approval score} \\ \textbf{in $\profile'$} \end{tabular} &  \begin{tabular}{@{}c@{}}\textbf{Approval score} \\ \textbf{in $\profile \concat \profile'$} \end{tabular} \\
				\midrule
				$p_1$ & $2$ & $10$ & $1$ & $11$ \\
				$p_2$ & $2$ & $1$ & $10$ & $11$ \\
				$p_3$ & $3$ & $9$ & $9$ & $18$ \\
				\bottomrule
			\end{tabular}
		\end{center}
		One can easily check that $\greedyApp$ would return $\{\{p_1, p_2\}\}$ on both $\profile$ and $\profile'$. However, on the joint profile $\profile \concat \profile'$, it would return $\{\{p_3\}\}$. Weak reinforcement is thus violated. The claim now follows from Lemma~\ref{lem:MLENecessaryCondition}.
	\end{proof}

	\section{Proportional PB Rules}
	
	A large part of recent research on PB has been devoted to the study of \emph{proportional} rules, i.e., rules that treat groups of agents fairly. In this section, we focus on the most prominent ones---Sequential Phragmén \citep{LCG22, BFLMP23} and approval-based variants of the Method of Equal Shares (MES) \citep{PPS21}---and show that they also cannot be interpreted as MLEs.
	
	\begin{definition}[Sequential Phragmén]
		Given an instance $I$ and a profile $\profile$, the Sequential Phragmén rule constructs budget allocations~$\pi$ using the following continuous process.
		
		Initially, $\pi$ is the empty set. Voters receive money in a virtual currency. They all start with a budget of~0 and that budget continuously increases as time passes. At time~$t$ a voter will have received an amount~$t$ of money. For any time $t$, let $P^\star_t$ be the set of projects $p \in \projSet$ for which the approvers together have at least $c(p)$ money available. As soon as, for a given $t$, the set $P^\star_t$ is non-empty, if there exists a $p \in P^\star_t$ such that $c(\pi \cup \{p\}) > b$, the process stops; otherwise one project from $P^\star_t$ is selected, the budget of its approvers is set to 0, and the process resumes.
		
		The outcome of Sequential Phragmén is the set of all budget allocations constructed by the procedure above (for all possible ways of breaking ties between the projects in $P^\star_t$).
	\end{definition}
	
	\noindent 
	Observe that, due to the fact that each voter's budget increases continuously, the approvers whose project~$p$ is selected at a given time~$t$ in fact will have accumulated \emph{exactly} $c(p)$, not more, which justifies setting their budgets back all the way to~0. Also note that with the above stopping condition---required to guarantee a property known as priceability \citep{LCG22}---Sequential Phragmén is not exhaustive. But it is exhaustive on unit-cost instances.
	
	Unfortunately, whatever assumptions we are willing to make regarding the noise model generating the approval sets of voters, Sequential Phragmén cannot be used to track the ground truth, not even for unit-cost instances.
	
	\begin{proposition}\label{prop:phragmen}
		There exists no noise model $\noiseModel$ such that Sequential Phragmén is the MLE for $\noiseModel$, not even on unit-cost instances with the additional assumption that the ground truth is exhaustive.
	\end{proposition}
	\begin{proof}
		Consider an instance $I$ with four projects denoted by $p_1$, $p_2$, $p_3$, and $p_4$, all of cost $1$, and budget limit of $b = 3$. We consider two profile, $\profile^1$ and $\profile^2$, as presented below, where $\profile^1$ is on the left and $\profile^2$ is on the right.

		\begin{center}
			\begin{tabular}{rcccc}
				\toprule
				& $p_1$ & $p_2$ & $p_3$ & $p_4$ \\
				\midrule
				Cost & 1 & 1 & 1 & 1 \\
				\midrule
				$A_1^1$ & \cmark & \xmark & \xmark & \xmark \\
				$A_2^1$ & \cmark & \xmark & \cmark & \cmark \\
				$A_3^1$ & \xmark & \cmark & \cmark & \cmark \\
				$A_4^1$ & \xmark & \cmark & \cmark & \cmark \\
				$A_5^1$ & \xmark & \cmark & \cmark & \cmark \\
				\midrule
				\multicolumn{5}{c}{$b = 3$} \\
				\bottomrule
			\end{tabular}
			\hspace{4em}
			\begin{tabular}{rcccc}
				\toprule
				& $p_1$ & $p_2$ & $p_3$ & $p_4$ \\
				\midrule
				Cost & 1 & 1 & 1 & 1 \\
				\midrule
				$A_1^2$ & \xmark & \cmark & \xmark & \xmark \\
				$A_2^2$ & \cmark & \xmark & \cmark & \xmark \\
				$A_3^2$ & \cmark & \xmark & \cmark & \xmark \\
				$A_4^2$ & \cmark & \xmark & \xmark & \cmark \\
				$A_5^2$ & \cmark & \xmark & \cmark & \cmark \\
				\midrule
				\multicolumn{5}{c}{$b = 3$} \\
				\bottomrule
			\end{tabular}
		\end{center}
		
		\noindent We claim that on both profile $\profile^1$ and profile $\profile^2$, Sequential Phragmén outputs a single budget allocation, namely $\pi = \{p_1, p_3, p_4\}$. Remember that, in the unit-cost setting, Sequential Phragmén is exhaustive. The budget allocation $\pi$ indeed is exhaustive.
		
		For $\profile^1$, after $\nicefrac{1}{2}$ units of money have been distributed, both $p_3$ and $p_4$ will have been bought at a price of $\nicefrac{1}{4}$ each. Once an additional $\nicefrac{1}{4}$ units of money have been injected, project $p_1$ is bought as well, making $\pi$ the unique budget allocation returned by Sequential Phragmén.
		
		For $\profile^2$, first $p_1$ is bought at a price of $\nicefrac{1}{4}$, then $p_3$ at a price of $\nicefrac{1}{3}$. Finally, after $\nicefrac{1}{3}$ additional units of money have been distributed, project $p_4$ is bought (at that time, the supporters of project $p_2$ have collected $\nicefrac{1}{4} + \nicefrac{2}{3} < 1$ money). The final outcome is thus indeed~$\{\pi\}$.
		
		Now, consider the joint profile $\profile^3 = \profile^1 \concat \profile^2$, and let us detail the run of Sequential Phragmén on $I$ and $\profile^3$. The first project to be bought is $p_3$ at price $\nicefrac{1}{7}$. Then, once an extra $\nicefrac{5}{42}$ units of money have been distributed, project $p_1$ is bought as well. At that time, the supporters of $p_4$ who do not approve of 
		$p_2$ have no money since they approve of $p_1$. On the other hand, the only supporter of $p_2$ who does not approve of $p_4$ has a strictly positive amount of money. Project $p_2$ will then be the last project selected (after another $\nicefrac{2}{21}$ units of money have been injected). Overall, the outcome will be $\{\{p_1, p_2, p_3\}\} \neq \{\pi\}$. 
		
		We have thus proven that the Sequential Phragmén rule fails weak reinforcement. Lemma~\ref{lem:MLENecessaryCondition} then concludes the proof. Note that all budget allocations we considered are exhaustive; the result thus also applies if we only focus on exhaustive ground truths.
	\end{proof}
	
	\noindent
	The same bad news apply to the MES-based rules. These rules are parametrised by a measure of the satisfaction of the voters. We call \emph{satisfaction function} (on singletons), any mapping from projects $p$ to satisfaction levels $\mu(p) \in \mathbb{R}_{>0}$.\footnote{Note that \citet{BFLMP23} give a more complete definition of satisfaction functions. Since we only need to discuss satisfaction of single projects here, we simplified the definition.} We can now define MES with respect to satisfaction function $\mu$.
	
	\begin{definition}[MES$_{\mu}$]
		Given an instance $I$, a profile $\profile = (A_1,\ldots,A_n)$ with $n$ agents, and a satisfaction function $\mu$, MES$_{\mu}$ constructs budget allocations $\pi$,
		initially empty, iteratively as follows.
		Every agent~$i$ is initially assigned a budget $b_i = \nicefrac{b}{n}$ of virtual money. Given a budget allocation $\pi$, a project $p \in \projSet \setminus \pi$ is said to be $\alpha$-affordable, for $\alpha \in \mathbb{R}_{\geq 0}$ if:
		\[\sum_{i \,\mid\, p \in A_i} \min(b_i, \alpha \cdot \mu(p)) \enspace \geq \enspace c(p).\]
		At a given round with current budget allocation $\pi$, if no project is $\alpha$-affordable for any
		$\alpha$, MES$_{\mu}$ terminates.
		Otherwise, let $P^\star$ be the set of projects that are $\alpha^*$-affordable for a minimum 
		$\alpha^*$. The rule selects one project $p \in P^\star$ ($\pi$ is updated to
		$\pi \cup \{p\}$), and every approver~$i$ of $p$ sees their budget reduced by $\min(b_i, \alpha \cdot \mu(p))$.
		
		The outcome of MES$_{\mu}$ is the set of all budget allocations constructed by the procedure above (for all possible ways of breaking ties between the projects in~$P^\star$).
	\end{definition}
	
	\noindent Note that MES$_\mu$ fails to be an exhaustive rule, for any satisfaction function $\mu$ and even on unit-cost instances. 
	
	We show that for no $\mu$ can MES$_{\mu}$ be interpreted as an MLE.
	
	\begin{proposition}
		For any given satisfaction function $\mu$, there is no noise model $\noiseModel$ such that MES$_{\mu}$ is the MLE for $\noiseModel$, not even on unit-cost instances.
	\end{proposition}
	
	\begin{proof}
		Consider an instance $I$ with two projects denoted by $p_1$ and $p_2$, both of cost 1, and a budget limit $b = 2$. Let $\mu$ be an arbitrary satisfaction function.
		
		Consider the two profiles $\profile^1 = (\{p_1\}, \{p_2\})$ and $\profile^2 = (\{p_1, p_2\}, \{p_1, p_2\})$. We claim that on both of these profiles, MES$_{\mu}$ would return $\pi = \{\{p_1, p_2\}\}$. Indeed, on $\profile^1$ both agents receive 1 unit of money and can both afford the project they approve of. On $\profile^2$ both agents approve of all the projects and can afford them. Now, for $\profile^3 = (\{p_1\}, \{p_2\}, \{p_1, p_2\}, \{p_1, p_2\}) = \profile^1 \concat \profile^2$, we claim that MES$_{\mu}$ would return either $\{\{p_1\}\}$, $\{\{p_2\}\}$, or $\{\{p_1\}, \{p_2\}\}$. Here the initial budget is $\nicefrac{1}{2}$ for each agent. Thus, the approvers of $p_1$ collectively have $\nicefrac{3}{2}$ units of money, and the same is true for $p_2$. Since $\mu(p) > 0$ for both $p_1$ and $p_2$ (by definition of satisfaction functions) and since we can choose $\alpha$ arbitrarily large, this implies that MES$_{\mu}$ would select either $p_1$ or $p_2$ in the first round. Let $p^\star$ be the selected project and $p$ the other project. To buy $p^\star$, all its approvers need to pay $\nicefrac{1}{3}$. The approvers of $p$ are thus now left with $\nicefrac{1}{2} + 2 \cdot (\nicefrac{1}{2} - \nicefrac{1}{3}) = \nicefrac{1}{2} + \nicefrac{1}{3} < 1$, which is not enough to afford $p$. There is thus no way for MES$_{\mu}$ to return $\{\{p_1, p_2\}\}$ on $\profile^3$.
		
		We have thus shown that MES$_{\mu}$ fails weak reinforcement. Lemma~\ref{lem:MLENecessaryCondition} then concludes the proof.
	\end{proof}
	
	\noindent
	Note that for the proofs of both the results in this section, the outcomes on the initial profiles are always disjoint from the outcomes on the joint profiles. This implies that even resolute versions of the rules (obtained by introducing some form of tie-breaking), or any refinement, would also fail weak reinforcement.

	\section{Monotonic Argmax Rules}
	\label{sec:Additive_Argmax_Rules}
	
	As we have seen, our first obstacle to finding rules that are MLEs is that none of the proportional rules we considered satisfy weak reinforcement. Aiming for more positive results, we will now follow a different approach: instead of checking whether known PB rules do satisfy weak reinforcement, we will start from rules we know satisfy it, and investigate their epistemic abilities. To this end, we will thus focus on \emph{monotonic argmax rules}, a large class of rules, all of which satisfy weak reinforcement.	
	We start by defining both what an argmax rule is and what makes such a rule monotonic.
	
	\begin{definition}[Monotonic Argmax Rules]
		A PB rule $\pbRule$ is called an \emph{argmax rule} if there exists a function $f$, taking as input an instance $I$, a profile $\profile$, and a budget allocation $\pi$, and returning a number $f(I, \profile, \pi) \in \mathbb{R}$, such that for all instances $I$ and all profiles $\profile$, we have:
		\[\pbRule(I, \profile) = \argmax_{\pi \in \allocSet(I)} f(I, \profile, \pi).\]
		An argmax rule defined via the function $f$ is called \emph{monotonic} if for every two profiles $\profile$ and $\profile'$ and every two budget allocations $\pi$ and $\pi'$, the following two conditions hold:
		\begin{enumerate}
			\item $\displaystyle \left.
			\begin{array}{r}
				f(I, \profile, \pi) < f(I, \profile, \pi')\\
				f(I, \profile', \pi) < f(I, \profile', \pi')
			\end{array}
			\right\} \Longrightarrow f(I, \profile \concat \profile', \pi) < f(I, \profile \concat \profile', \pi'),$
			\item $\displaystyle \left.
			\begin{array}{r}
				f(I, \profile, \pi) = f(I, \profile, \pi')\\
				f(I, \profile', \pi) = f(I, \profile', \pi')
			\end{array}
			\right\} \Longrightarrow f(I, \profile \concat \profile', \pi) = f(I, \profile \concat \profile', \pi').$
		\end{enumerate}
	\end{definition}
	
	\noindent Note that \emph{every} rule is an argmax rule---$f$ can simply be the indicator function on the outcome of the rule for a given instance and profile---but not all rules are monotonic.
	
	Are the monotonic argmax rules good candidates for being MLEs? Yes, they are, as we can show that monotonic argmax rules all satisfy weak reinforcement.
	
	\begin{proposition}
		\label{prop:ArgmaxRules_WeakReinforcement}
		Every monotonic argmax rule satisfies weak reinforcement.
	\end{proposition}
	
	\begin{proof}
		Consider the monotonic argmax rule $\pbRule$ defined by the function $f$. Let $I$ be an instance, and $\profile$ and $\profile'$ two profiles over $I$ such that $\pbRule(I, \profile) = \pbRule(I, \profile')$. We show that we also have $\pbRule(I, \profile \concat \profile') = \pbRule(I, \profile)$.
		
		Consider a budget allocation $\pi \in \pbRule(I, \profile)$. Since $\pbRule$ is an argmax rule, for all $\pi' \in \allocSet(I) \setminus \pbRule(I, \profile)$, we know that $f(I, \profile, \pi) > f(I, \profile, \pi')$. This also holds for $\profile'$. By the definition of a monotonic argmax rule, we immediately get that $f(I, \profile \concat \profile', \pi) > f(I, \profile \concat \profile', \pi')$.
		
		Moreover, for any two budget allocations $\pi, \pi' \in \pbRule(I, \profile)$, we have $f(I, \profile, \pi) = f(I, \profile, \pi')$. The same also holds for $\profile'$. Since $\pbRule$ is monotonic, we thus immediately get that $f(I, \profile \concat \profile', \pi) = f(I, \profile \concat \profile', \pi')$.
		
		Overall, we proved that $(i)$ no budget allocation that was winning under $\profile$ or $\profile'$ is dominated under $\profile \concat \profile'$, and $(ii)$ that all budget allocations that were winning under $\profile$ and $\profile'$ all have the same score. It is then immediate that $\pbRule(I, \profile \concat \profile') = \pbRule(I, \profile)$.
	\end{proof}
	
	\noindent In what follows we will introduce and study several concrete examples of monotonic argmax rules, based either on the Nash or on the utilitarian social welfare.
	
	\subsection{Nash Social Welfare}
	\label{subsec:Nash_Social_Welfare}
	
	We first study rules based on the Nash social welfare. It tries to reach balanced outcomes by measuring the score of a budget allocation as the product of the agents' levels of satisfaction.
	The concept of Nash social welfare provides strong guarantees in the context of the fair allocation of indivisible items \citep{CKM+19}. It has also been identified as an appealing rule for PB \citep{RoTa21}.
	
	\subsubsection{Cardinality and Cost Satisfaction}
	
	We first consider two common measures of satisfaction: based on the cardinality and on the cost of approved and selected projects \citep{TaFa19,ReMa23}. This gives rise to two argmax rules, $\appNashMaximizing$ and $\costAppNashMaximizing$, defined via the following functions:
	\begin{align*}
		f^N_{\mathit{app}}(I, \profile, \pi) & = 
		\begin{cases}
			\sum_{A \in \profile} \log(|A \cap \pi|) & \textit{if for all } A' \in \profile, |A' \cap \pi| \neq 0, \\
			0 & \textit{otherwise.}
		\end{cases} \\
		f^N_{\mathit{cost}}(I, \profile, \pi) & = 
		\begin{cases}
			\sum_{A \in \profile} \log(c(A \cap \pi)) & \textit{if for all } A' \in \profile, c(A' \cap \pi) \neq 0, \\
			0 & \textit{otherwise.}
		\end{cases}.
	\end{align*}
	
	\noindent It can be checked that these two argmax rules are indeed monotonic. From Proposition~\ref{prop:ArgmaxRules_WeakReinforcement}, we thus know that they satisfy weak reinforcement. Note that one would need to take the exponential of these function to make the Nash social welfare appear.
	
	Interestingly, under the common and natural assumption that all projects are approved by at least one agent, these two rules are exhaustive.
	
	We start our investigation by introducing a noise model, denoted by $\noiseMmaxNashAppCost$, for which for all $I$, $A$ and $\truth$, we have:
	\[\prob_{\noiseMmaxNashAppCost} (A \mid \truth, I) = \frac{1}{\noiseMmaxNashAppCostNORM{\truth}} c(A \cap \truth),\]
	where, $\noiseMmaxNashAppCostNORM{\truth}$ is a suitable normalisation factor ensuring that $\noiseMmaxNashAppCost$ is a well-defined probability distribution.
	
	Under this noise model, the probability of generating a given ballot~$A$ increases with the cost of the ground-truth projects in~$A$. The intuition here is that voters may reflect more carefully on expensive projects and thus are more likely to make correct choices for them. Moreover, the probability of generating~$A$ increases linearly in the ``quality'' of~$A$. How realistic this is, of course, is open to debate. On the one hand, this avoids having to assume extremely high probabilities for correctly identifying particularly expensive projects (in the case where the relationship would not be linear). On the other hand, the probability of generating a ballot that is completely wrong (a ballot not including even a single ground-truth project) is zero.
	
	Under $\noiseMmaxNashAppCost$, maximising the likelihood would be
	similar to maximising the cost-approval Nash social welfare of a budget allocation. However, for this intuitive connection to hold, it should be the case that the normalisation factor $\noiseMmaxNashAppCostNORM{\truth}$ is independent of $\truth$. So let us look at it in more detail.
	
	\begin{lemma}
		\label{lem:NormFact_MaxNashCost}
		For the noise model $\noiseMmaxNashAppCost$ to be a well-defined probability distribution, it must be the case that:
		\[\noiseMmaxNashAppCostNORM{\truth} = 2^{|\projSet| - 1} c(\truth).\]
	\end{lemma}
	
	\begin{proof}
		Consider an instance $I = \tuple{\projSet, c, b}$, an approval ballot $A \subseteq \projSet$, and a ground truth $\truth \in \allocSet(I)$. For $\noiseMmaxNashAppCost$ to be a probability distribution, it should be the case that:
		\[\sum_{A \subseteq \projSet} \prob_{\noiseMmaxNashAppCost} (A \mid \truth, I) = 1 \quad \Longleftrightarrow \quad \noiseMmaxNashAppCostNORM{\truth} = \sum_{A \subseteq \projSet} c(|A \cap \truth|).\]
		Remember that there are $2^{|\projSet|}$ subsets of projects and that any project $p \in \projSet$ appears in exactly half of them. Each time a project $p \in \truth$ appears in a subset $A \subseteq \projSet$, its contribution to the value of $\noiseMmaxNashAppCostNORM{\truth}$ is exactly $c(p)$. We thus have:
		\[\noiseMmaxNashAppCostNORM{\truth} = \sum_{A \subseteq \projSet} c(|A \cap \truth|) = \sum_{p \in \truth} 2^{|\projSet| - 1} c(p) = 2^{|\projSet| - 1} c(\truth). \qedhere\] 
	\end{proof}
	
	\noindent This result tells us that the normalisation factor of the noise model $\noiseMmaxNashAppCost$ depends on the ground truth, meaning that the value of the likelihood is impacted by the ground truth one is considering when computing the MLE. In particular, we cannot conclude that the Nash cost-approval maximising rule is the MLE for this noise model since not all feasible budget allocations have the same cost.
	
	Are there specific cases for which the normalisation factor is independent of the ground truth? Yes, namely for unit-cost instances, as then all exhaustive allocations have the same cost.
	
	\begin{proposition}
		\label{prop:NashAppMax_MLE_Exhaustive}
		Under the assumption that the ground truth is exhaustive, both $\appNashMaximizing$ and $\costAppNashMaximizing$ are the MLE of the noise model $\noiseMmaxNashAppCost$ for unit-cost instances.
	\end{proposition}
	
	\begin{proof}
		Let $I$ be a unit-cost instance $I$. Consider any two exhaustive budget allocations $\pi, \pi' \in \allocSetEx(I)$. Since we have $|\pi| = |\pi'| = c(\pi) = c(\pi')$, Lemma~\ref{lem:NormFact_MaxNashCost} entails that $\noiseMmaxNashAppCostNORM{\pi} = \noiseMmaxNashAppCostNORM{\pi'}$. For any profile $\profile$, we have then:
		\begin{align*}
			\argmax_{\pi \in \allocSetEx(I)} L_{\noiseMmaxNashAppCost}(\profile, \pi, I) & = \argmax_{\pi \in \allocSetEx(I)} \prod_{A \in \profile} \frac{c(A \cap \pi)}{\noiseMmaxNashAppCostNORM{\pi}} \\
			& = \argmax_{\pi \in \allocSetEx(I)} \prod_{A \in \profile} c(A \cap \pi) \\
			& = \costAppNashMaximizing(I, \profile).
		\end{align*}
		The last line follows form the fact that $\costAppNashMaximizing$ is exhaustive.
		
		Given that $\appNashMaximizing$ and $\costAppNashMaximizing$ coincide on unit-cost instances, this also applies to $\appNashMaximizing$.
	\end{proof}
	
	\noindent The fact that $\appNashMaximizing$ and $\costAppNashMaximizing$ are MLEs for $\noiseMmaxNashAppCost$ only under some restricted hypothesis is the first hint of a general impossibility result. Indeed, we can show that there are no noise models for which these rule are MLEs.
	
	\begin{theorem}
		\label{thm:NashAppMax_NotMLE_UnitCost}
		There is no noise model $\noiseModel$ such that either $\appNashMaximizing$ or $\costAppNashMaximizing$ is the MLE of $\noiseModel$, not even on unit-cost instances.
	\end{theorem}
	
	\begin{proof}
		Consider an instance $I$ with two projects $p_1$ and $p_2$ of cost 1, and a budget limit of $b = 2$. Let $\noiseModel$ be a generic noise model, and denote by $\prob_A^\pi$ the value of $\prob_\noiseModel(A \mid \pi, I)$ for any $A$ and $\pi$. To simplify notation, we omit braces around sets. 
		
		For the noise model $\noiseModel$ to be a well-defined probability distribution, the following two equalities should be satisfied:
		\begin{align}
			\sum_{A \subseteq \projSet} \prob_{A}^{p_1} = 1 &&\enspace\Leftrightarrow\enspace&& \prob_{\emptyset}^{~p_1} &&+&& \prob_{p_1}^{~p_1} &&+&& \prob_{p_2}^{~p_1} &&+&& \prob_{p_1, p_2}^{~p_1} &&=&& 1, \label{eq:proofNotMLEMaxAppLine1} \\
			\sum_{A \subseteq \projSet} \prob_{A}^{p_1, p_2} = 1 &&\enspace\Leftrightarrow\enspace&& \prob_{\emptyset}^{~p_1, p_2} &&+&& \prob_{p_1}^{~p_1, p_2} &&+&& \prob_{p_2}^{~p_1, p_2} &&+&& \prob_{p_1, p_2}^{~p_1, p_2} &&=&& 1. \label{eq:proofNotMLEMaxAppLine2}
		\end{align}
		Now, on the single-agent profile $\profile = (\emptyset)$, $\appNashMaximizing$ returns $\allocSet(I)$. So for $\appNashMaximizing$ to be the MLE of $\noiseModel$, we must have $\prob_{\emptyset}^{~p_1} = \prob_{\emptyset}^{~p_1, p_2}$. Moreover, on $\profile = (\{p_1\})$, we have $\appNashMaximizing(I, \profile) = \{\{p_1\}, \{p_1, p_2\}\}$, so $\prob_{p_1}^{~p_1} = \prob_{p_1}^{~p_1, p_2}$. Using these two equalities and by subtracting \eqref{eq:proofNotMLEMaxAppLine2} from \eqref{eq:proofNotMLEMaxAppLine1}, we get:
		\begin{equation}
			(\prob_{p_2}^{~p_1} - \prob_{p_2}^{~p_1, p_2}) \quad + \quad (\prob_{p_1, p_2}^{~p_1} - \prob_{p_1, p_2}^{~p_1, p_2}) \quad = \quad 0. \label{eq:proofNotMLEMaxAppLine3}
		\end{equation}
		Now, since $\appNashMaximizing(I, (\{p_2\})) = \{\{p_2\}, \{p_1, p_2\}\}$, we must have $\prob_{p_2}^{~p_1, p_2} > \prob_{p_2}^{~p_1}$. For $\profile = (\{p_1, p_2\})$, we have $\appNashMaximizing(I, \profile) = \{\{p_1, p_2\}\}$. We can then derive $\prob_{p_1, p_2}^{~p_1, p_2} > \prob_{p_1, p_2}^{~p_1}$. These two last inequalities contradict \eqref{eq:proofNotMLEMaxAppLine3}. It is then impossible for $\appNashMaximizing$ to be the MLE of $\noiseModel$ on $I$. From the unit-cost assumption, it is clear that this also applies to $\costAppNashMaximizing$.
	\end{proof}
	
	\noindent This impossibility result concludes this part of our analysis. We will now consider ``normalised'' satisfaction functions.
	
	\subsubsection{Normalised Satisfaction}
	\label{subsubsec:Normalised_Satisfaction}
	
	In the hope of overcoming Theorem~\ref{thm:NashAppMax_NotMLE_UnitCost}, we also consider normalised variants of the cardinality and cost satisfaction functions. In these variants, the satisfaction of an agent is expressed in terms of the proportion of the outcome that satisfies them. The rules induced by these normalised satisfaction functions are denoted by $\relAppNashMaximizing$ and $\relCostAppNashMaximizing$, and they are the argmax rules defined in terms of the following functions:
	\begin{align*}
		\tilde{f}^N_{\mathit{app}}(I, \profile, \pi) & = 
		\begin{cases}
			\sum_{A \in \profile} \log(|A \cap \pi|) - \log(|\pi|) & \textit{if for all } A' \in \profile, |A' \cap \pi| \neq 0, \\
			0 & \textit{otherwise.}
		\end{cases} \\
		\tilde{f}^N_{\mathit{cost}}(I, \profile, \pi) & = 
		\begin{cases}
			\sum_{A \in \profile} \log(c(A \cap \pi)) - \log(c(\pi)) & \textit{if for all } A' \in \profile, c(A' \cap \pi) \neq 0, \\
			0 & \textit{otherwise.}
		\end{cases}.
	\end{align*}
	
	\noindent These rules are inspired by the concept of \emph{relative satisfaction} introduced by \citet{LMR21}. However, while in the cited work the authors normalised the satisfaction with respect to the ballot, we define it here with respect to the budget allocation.
	
	It is worth noting that the rules $\relAppNashMaximizing$ and $\relCostAppNashMaximizing$ can lead to extreme behaviours. For example, consider an instance with budget limit $b$, that we assume to be even, and a set of projects $\projSet = \{p^\star\} \cup \{p_1, \ldots, p_b\}$ of arbitrary cost lower than $b$. Consider the two-agent profile $\profile$ such that:
	\begin{align*}
		A_1 & = \{p^\star\} \cup \{p_1, p_3, \ldots, p_{b - 1}\} & A_2 & = \{p^\star\} \cup \{p_2, p_4, \ldots, p_b\}.
	\end{align*}
	Then, according to both rules, selecting just $p^\star$ is better than anything else. Even if this can seem extreme, these rules can still be justified when considering voters who would rather save public money than use it on projects they do not approve. This would correspond to associating a strong rejection, rather than indifference, with the action of not approving a project. 
	
	Note that this example also implies that the rules are not exhaustive, even on unit-cost instances.
	
	Let us first investigate the rule $\relCostAppNashMaximizing$. We will continue using the noise model $\noiseMmaxNashAppCost$ introduced earlier.
	
	Recall the expression we found for the normalisation factor $\noiseMmaxNashAppCostNORM{\truth}$ in Lemma~\ref{lem:NormFact_MaxNashCost}. Plugging it into the definition of $\noiseMmaxNashAppCost$, we obtain the following expression for any instance $I$, approval ballot $A$, and ground truth $\truth$:
	\[\prob_{\noiseMmaxNashAppCost} (A \mid \truth, I) = \frac{1}{2^{|\projSet| - 1}} \frac{c(A \cap \truth)}{c(\truth)}.\]
	
	\noindent From this, it should be clear that $\relCostAppNashMaximizing$ is the MLE of $\noiseMmaxNashAppCost$.
	
	\begin{theorem}
		\label{thm:NashNormApp_MLE}
		The rule $\relCostAppNashMaximizing$ is the MLE of the noise model $\noiseMmaxNashAppCost$.
	\end{theorem}
	
	\begin{proof}
		Let $I = \tuple{\projSet, c, b}$ be an instance. The likelihood of a profile $\profile$ and a budget allocation $\pi \in \allocSet(I)$ under the noise model $\noiseMmaxNashAppCost$ is:
		\[L_{\noiseMmaxNashAppCost}(\profile, \pi, I) = \prod_{A \in \profile} \frac{1}{2^{|\projSet| - 1}} \frac{c(A \cap \pi)}{c(\pi)} = \left(\frac{1}{2^{|\projSet| - 1}}\right)^{|\profile|}\prod_{A \in \profile} \frac{c(A \cap \pi)}{c(\pi)}.\]
		Since the first multiplicative factor in the above expression is constant over all budget allocations, we have:
		\[\argmax_{\pi \in \allocSet(I)} L_{\noiseMmaxNashAppCost}(\profile, \pi, I) = \argmax_{\pi \in \allocSet(I)} \prod_{A \in \profile}\frac{c(A \cap \pi)}{c(\pi)} = \relCostAppNashMaximizing(I, \profile).\]
		Thus, maximising the likelihood under $\noiseMmaxNashAppCost$ is the same as maximising the social welfare in the sense of $\relCostAppNashMaximizing$.
	\end{proof}
	
	\noindent We have finally been able to find a PB rule that can be interpreted as an MLE. In the following we will show a similar result for $\relAppNashMaximizing$. For this rule we introduce a new noise model denoted by $\noiseMmaxNashApp$. It is such that for any instance $I$, approval ballot $A$, and ground truth $\truth$, we have:
	\[\prob_{\noiseMmaxNashApp} (A \mid \truth, I) = \frac{1}{\noiseMmaxNashAppNORM{\truth}} |A \cap \truth|,\]
	where $\noiseMmaxNashAppNORM{\truth}$ is a suitable normalisation factor.
	
	Using a similar proof technique as above, we show that $\relAppNashMaximizing$ is the MLE of $\noiseMmaxNashApp$.
	
	\begin{theorem}
		\label{thm:NashNormCost_MLE}
		The rule $\relAppNashMaximizing$ is the MLE of the noise model $\noiseMmaxNashApp$.
	\end{theorem}
	
	\begin{proof}
		Let us first compute the exact value of the normalisation factor $\noiseMmaxNashAppNORM{\truth}$. For the noise model $\noiseMmaxNashApp$ to be a well-defined probability distribution, the following must hold:
		\[\sum_{A \subseteq \projSet} \prob_{\noiseMmaxNashApp} (A \mid \truth, I) = 1 \enspace \Leftrightarrow \enspace \noiseMmaxNashAppNORM{\truth} = \sum_{A \subseteq \projSet} |A \cap \truth| \enspace \Leftrightarrow \enspace \noiseMmaxNashAppNORM{\truth} = 2^{|\projSet| - 1} |\truth|.\]
		Hence, given a profile $\profile$, we have:
		\begin{align*}
			\argmax_{\pi \in \allocSet(I)} L_{\noiseMmaxNashApp}(\profile, \pi) & = \argmax_{\pi \in \allocSet(I)} \prod_{A \in \profile} \frac{1}{2^{|\projSet| - 1}} \frac{|A \cap \pi|}{|\pi|} \\
			& = \argmax_{\pi \in \allocSet(I)} \prod_{A \in \profile} \frac{|A \cap \pi|}{|\pi|} \\
			& = \relAppNashMaximizing(I, \profile).
		\end{align*}
		This immediately implies that $\relAppNashMaximizing$ is the MLE of $\noiseMmaxNashApp$.
	\end{proof}
	
	\noindent This concludes our study of the Nash social welfare. We now turn to the more standard measure of the social welfare, the utilitarian social welfare.
	
	\subsection{Utilitarian Social Welfare}
	\label{subsec:Utilitarian_Social_Welfare}
	
	Let us now turn to the analysis of monotonic argmax rules defined in terms of utilitarian social welfare. As we did before, we will consider different satisfaction functions.
	
	\subsubsection{Cardinality and Cost Satisfaction}
	
	Following \citet{TaFa19}, we define the approval maximising rule $\appMaximizing$ and the cost-approval maximising rule $\costAppMaximizing$ as the argmax rules determined by $f_{\mathit{app}}$ and $f_{\mathit{cost}}$, respectively:
	\begin{align*}
		f_{\mathit{app}}(I, \profile, \pi) & = \sum_{A \in \profile} |A \cap \pi|, \\
		f_{\mathit{cost}}(I, \profile, \pi) & = \sum_{A \in \profile} c(A \cap \pi).
	\end{align*}
	
	\noindent It should be clear that these two rules are monotonic argmax rule (since the score they use is additive) and thus both satisfy weak reinforcement. Note also that these two rules are exhaustive.
	
	Following an idea developed by \citet{CoSa05} for scoring rules (in the standard voting framework), we introduce the noise model $\noiseMmaxapp$. It is defined such that for any instance $I = \tuple{\projSet, c, b}$, approval ballot $A \subseteq \projSet$, and ground truth $\truth \in \allocSet(I)$:
	\[\prob_{\noiseMmaxapp}(A \mid \truth, I) = \frac{1}{\noiseMmaxappNORM{\truth}} \prod_{p \in \projSet}2^{\mathds{1}_{p \in A \cap \truth}} = \frac{1}{\noiseMmaxappNORM{\truth}}2^{|A \cap \truth|},\]
	where $\noiseMmaxappNORM{\truth}$ is a suitable normalisation factor.
	
	$\noiseMmaxapp$ is a particularly simple manifestation of what we would expect to see in a noise model: any possible ballot might be generated in principle, but the probability of generating ballot~$A$ increases exponentially with the size of the intersection between $A$ and the ground truth.
	
	With this noise model, maximising the likelihood may appear to have the same effect as maximising the approval score of a budget allocation. It could then be the case that the approval maximising rule is the MLE of $\noiseMmaxapp$. However, for this to hold, one has to have a closer look at the normalisation factor.
	
	\begin{lemma}\label{lem:Zapp}
		For the noise model $\noiseMmaxapp$ to be a well-defined probability distribution, it must be the case that:
		\[\noiseMmaxappNORM{\truth} = 2^{|\projSet|} \left(\frac{3}{2}\right)^{|\truth|}.\]
	\end{lemma}
	
	\begin{proof}
		Consider any instance $I = \tuple{\projSet, c, b}$. Let $A \subseteq \projSet$ be an approval ballot and $\truth \in \allocSet(I)$ a ground truth. For $\noiseMmaxapp$ to be a probability distribution, it should be the case that:
		\[\sum_{A \subseteq \projSet} \prob_{\noiseMmaxapp}(A \mid \truth, I) = 1 \qquad \Longleftrightarrow \qquad \noiseMmaxappNORM{\truth} = \sum_{A \subseteq \projSet} 2^{|A \cap \truth|}.\]
		Let's do some combinatorics. For $k \in \{0, \ldots, |\truth|\}$, how many subsets of $\projSet$ will intersect with $\truth$ on exactly $k$ projects? A suitable subset will consist of $k$ projects from $\truth$ that make up the intersection and any number $j \in \{0, \ldots, |\projSet| - |\truth|\}$ of projects from $\projSet \setminus \truth$ that do not have any impact on the intersection. Each such subset of projects contributes $2^k$ to the value of $\noiseMmaxappNORM{\truth}$. We thus have:
		\allowdisplaybreaks
		\begin{align*}
			\noiseMmaxappNORM{\truth} & = \sum_{k = 0}^{|\truth|} 2^k \sum_{j = 0}^{|\projSet| - |\truth|} \binom{|\truth|}{k} \binom{|\projSet| - |\truth|}{j} \\
			& = \sum_{k = 0}^{|\truth|} \binom{|\truth|}{k} 2^k \sum_{j = 0}^{|\projSet| - |\truth|} \binom{|\projSet| - |\truth|}{j} \\
			& = 2^{|\projSet| - |\truth|} \sum_{k = 0}^{|\truth|} \binom{|\truth|}{k} 2^k \\
			& = 2^{|\projSet|} \left(\frac{3}{2}\right)^{|\truth|},
		\end{align*}
		where the last two lines are derived from the binomial expansion.
	\end{proof}
	
	\noindent The normalisation factor of $\noiseMmaxapp$ thus depends on the ground truth, since not all feasible budget allocations have the same cardinality. We thus cannot conclude that the approval maximising rule is the MLE of this noise model.
	
	Interestingly, this is not the case on unit-cost instances when considering exhaustive budget allocations.
	
	\begin{proposition}
		\label{prop:appMax_MLE_Exhaustive}
		Under the assumption that the ground truth is exhaustive, both $\appMaximizing$ and $\costAppMaximizing$ are the MLE of the noise model $\noiseMmaxapp$ for unit-cost instances.
	\end{proposition}
	
	\begin{proof}
		Let $I$ be a unit-cost instance. For any two exhaustive budget allocations $\pi$ and $\pi' \in \allocSetEx(I)$, by virtue of Lemma~\ref{lem:Zapp}, we have $\noiseMmaxappNORM{\pi} = \noiseMmaxappNORM{\pi'}$. 
		So, for any profile~$\profile$, we have:
		\begin{align*}
			\argmax_{\pi \in \allocSetEx(I)} L_{\noiseMmaxapp}(\profile, \pi, I) & = \argmax_{\pi \in \allocSetEx(I)} \prod_{A \in \profile} \frac{1}{\noiseMmaxappNORM{\pi}} 2^{|A \cap \pi|} \\
			& = \argmax_{\pi \in \allocSetEx(I)} 2^{\sum_{A \in \profile} |A \cap \pi|} \\
			& = \argmax_{\pi \in \allocSetEx(I)} \sum_{A \in \profile} |A \cap \pi| \\
			& = \appMaximizing(I, \profile).
		\end{align*}
		The last line follows from the fact that $\appMaximizing$ is exhaustive.
		
		$\appMaximizing$ thus coincides with the MLE on $I$ for the noise model $\noiseMmaxapp$. Moreover, since $\appMaximizing$ and $\costAppMaximizing$ coincide on unit-cost instances, the result also applies to $\costAppMaximizing$.
	\end{proof}
	
	\noindent But this result is only half satisfactory. Can we find an impossibility result similar to the one we had for $\appNashMaximizing$ and $\costAppNashMaximizing$? It is actually easy to see that the proof we gave for Theorem~\ref{thm:NashAppMax_NotMLE_UnitCost} also works for both $\appMaximizing$ and~$\costAppMaximizing$.
	
	\begin{theorem}
		\label{thm:appMax_NotMLE_UnitCost}
		There is no noise model $\noiseModel$ such that either $\appMaximizing$ or $\costAppMaximizing$ is the MLE of $\noiseModel$, not even on unit-cost instances.
	\end{theorem}
	
	\begin{proof}
		Consider the instance $I$ used in the proof of Theorem~\ref{thm:NashAppMax_NotMLE_UnitCost}. We claim that for all profiles that are relevant for the proof, $\appMaximizing$ and $\appNashMaximizing$ coincide. We list them below.
		\begin{gather*}
			\appMaximizing(I, (\{p_1\})) = \{\{p_1\}, \{p_1, p_2\}\} = \appNashMaximizing(I, (\{p_1\})). \\
			\appMaximizing(I, (\{p_2\})) = \{\{p_2\}, \{p_1, p_2\}\} = \appNashMaximizing(I, (\{p_2\})). \\
			\appMaximizing(I, (\{p_1, p_2\})) = \{\{p_1, p_2\}\} = \appNashMaximizing(I, (\{p_1, p_2\})). \\
			\appMaximizing(I, (\emptyset)) = \allocSet(I) = \appNashMaximizing(I, (\emptyset)).
		\end{gather*}
		Given that on unit-cost instances $\appMaximizing$ and $\costAppMaximizing$ coincide, this also applies to $\costAppMaximizing$.
	\end{proof}
	
	\noindent 
	It is interesting to note that the same result applies if one maximises the total relative satisfaction as introduced by \citet{LMR21}, i.e., the sum over all agents of the number (or cost) of approved and selected projects divided by the number (or cost) of approved projects. Indeed, one can easily check that for all relevant profiles this maximisation rule coincides with $\appNashMaximizing$. This trivially also applies when using the Nash social welfare with relative satisfaction.
	
	As an aside, observe that the greedy cost approval rule $\greedyApp$ discussed in Section~\ref{sec:truthtracking} coincides with $\appMaximizing$ on unit-cost instances. Thus, both Proposition~\ref{prop:appMax_MLE_Exhaustive} and Theorem~\ref{thm:appMax_NotMLE_UnitCost} apply to $\greedyApp$ as well, and we can refine Proposition~\ref{prop:greedy} as follows.
	
	\begin{corollary}
		\label{cor:GreedyRule}
		Under the assumption that the ground truth is exhaustive, the greedy cost approval rule $\greedyApp$ is the MLE for $\noiseMmaxapp$ for unit-cost instances. 
		
		Moreover, for unconstrained ground truths, there is no noise model $\noiseModel$ such that $\greedyApp$ is the MLE for $\noiseModel$, not even on unit-cost instances.
	\end{corollary}
	
	\noindent
	These insights apply also for any refinements of these rules, such as the \emph{leximax rule} introduced by \citet{REH20}.

	\subsubsection{Normalised Satisfaction}
	
	Let us conclude our formal analysis by briefly mentioning the utilitarian social welfare with the normalised satisfaction functions. For the same reasons as for $\relAppNashMaximizing$ and $\relCostAppNashMaximizing$, the corresponding utilitarian rules are not exhaustive. Analysing the epistemic status of these rules however turns out to be rather intricate, even on unit-cost instances. Indeed, it is less clear what a suitable noise model might look like, especially due to the complications related to the potential normalisation factor. Exploring these rules remains an interesting open problem.
	
	\section{Conclusion}
	\label{sec:conclusion}
	
	We have initiated the study of PB through the truth-tracking lens. For a total of eleven rules, we investigated whether they can be interpreted as MLEs. For those that cannot, we tried to identify specific conditions under which they still can serve as MLEs. All our results are summarised in Table~\ref{tab:results}.
	
	There is still some work to be done regarding the study of MLEs in the context of PB. The most obvious task would be to answer the open questions corresponding to the two missing cells in Table~\ref{tab:results}. In the case of our positive results, one could try to come up with more natural noise models for which the rules are MLEs, e.g., noise models for which not all the supersets of the ground truth have the same probability. Our work also shows a certain tension between efficiency requirements (exhaustiveness) and truth-tracking ability: the two rules that we proved to be MLEs (for the general case) both fail exhaustiveness. This interaction deserves further study. Similarly, we found that the most studied rules that enforce proportional representation fail to satisfy weak reinforcement. It would be interesting to investigate whether or not this constitutes a formal incompatibility.
	
	Then, on top of the MLE concept, the epistemic approach also offers several other ways of studying voting rules.
	As we have seen, finding rules that are MLEs is not an easy task. It might be the case that this is simply too demanding a requirement.
	Instead, other criteria that have been studied in the literature on epistemic social choice could be applied to the PB setting as well.
	For instance, it could be interesting to study PB rules with respect to their sample complexity \citep{CPS16} or their robustness against noise \citep{CKKK22}---a criterion that is somewhat similar to the MLE requirement but easier to satisfy. All of these constitute interesting directions for future work on a topic that is still very much under-studied and deserving of further attention.

	\begin{table*}
		\centering
		\resizebox{\linewidth}{!}{
			\begin{tabular}{lccccccccccc}
				\toprule
				& && \textbf{Greedy} & \multicolumn{4}{c}{\textbf{Approval-max}} & \multicolumn{4}{c}{\textbf{Cost-approval-max}} \\
				& \textbf{Sequential} & \textbf{MES$_{\mu}$} & \textbf{Cost} & \multicolumn{2}{c}{Standard} & \multicolumn{2}{c}{Normalised} &  \multicolumn{2}{c}{Standard} & \multicolumn{2}{c}{Normalised} \\
				& \textbf{Phragmén} & (for all $\mu$) & \textbf{Approval} & $\sum$ & $\prod$ & $\sum$ & $\prod$ & $\sum$ & $\prod$ & $\sum$ & $\prod$ \\
				\midrule
				Unit-cost$_{\mathit{EX}}$ & \xmark & -- & \cmark & \cmark & \cmark & -- & -- & \cmark & \cmark & -- & -- \\
				Unit-cost & \xmark & \xmark & \xmark & \xmark & \xmark & ? & \cmark & \xmark & \xmark & ? & \cmark \\
				General case & \xmark & \xmark & \xmark & \xmark & \xmark & ? & \cmark & \xmark & \xmark & ? & \cmark \\
				\bottomrule
			\end{tabular}
		}
		\caption{Summary of results: The sum $\Sigma$ and product~$\Pi$ symbols represent the utilitarian and the Nash variant of a welfare-based rule. A check-mark~\cmark{} indicates that there exists a noise model for which the rule is an MLE and a cross-mark~\xmark{} the fact that it is impossible to find such a noise model. The\ $_{\mathit{EX}}$ subscript signifies that we make the additional assumption that the ground truth is exhaustive. This assumption would not be meaningful for non-exhaustive rules; this is indicated by a bar~--. Remember that Sequential Phragmén is exhaustive on unit-cost instances. Question marks~? indicate open problems.}
		\label{tab:results}
	\end{table*}
	
	\paragraph{Acknowledgments.} We are grateful to multiple anonymous reviewers for the valuable feedback we received on this work.
	
	\bibliographystyle{ACM-Reference-Format}
	\bibliography{abb,PB}
	
\end{document}